\documentclass[10pt]{article}

\usepackage{amsmath,amsthm,amssymb,dsfont,paralist}
\usepackage{graphicx}

\newtheorem{df}{Definition}[section]
\newtheorem{lm}[df]{Lemma}
\newtheorem{ex}[df]{Example}

\newtheorem{theo}[df]{Theorem}
\newtheorem{cor}[df]{Corollary}
\newtheorem{ob}[df]{Observation}

\newcommand{\eat}[1]{}

\newcommand{\Val}{\mathrm{val}}
\newcommand{\supp}{supp}
\newcommand{\im}{im}

\newcommand{\PDA}{\mathrm{PDA}}
\newcommand{\uPDA}{\mathrm{uPDA}}

\newcommand{\CF}{\mathrm{CF}}
\newcommand{\uCF}{\mathrm{uCF}}

\newcommand{\fl}{\langle \! \langle}
\newcommand{\wt}{\mathrm{wt}}
\newcommand{\lab}{\mathrm{lab}}\newcommand{\fr}{\rangle \! \rangle}
\newcommand{\one}{\mathds{1}}
\newcommand{\G}{{\cal G}}
\newcommand{\M}{{\cal M}}
\newcommand{\D}{\mathrm{D}}
\newcommand{\HH}{{\cal H}}
\newcommand{\beh}{|\!|}
\newcommand{\nat}{\mathbb{N}}

\begin{document}
\begin{center}

{\LARGE The Chomsky-Sch\"utzenberger Theorem\\[3mm] for Quantitative Context-Free Languages}

\

\

\begin{tabular}{cc}
{\large Manfred Droste} & {\large Heiko Vogler}\\[1mm]
{\small {\it Institute of Computer Science}}
& {\small {\it Department of Computer Science}}\\[1mm]
{\small {\it Leipzig University}} & {\small {\it Technische
    Universit\"at Dresden}}\\[1mm]
{\small {\it D-04109 Leipzig, Germany}} & {\small {\it D-01062 Dresden, Germany}} \\[1mm] 
{\small \tt droste@informatik.uni-leipzig.de} &
{\small \tt Heiko.Vogler@tu-dresden.de}
\end{tabular}

\

\today

\end{center}

\hyphenation{semi-rings}

\begin{abstract}
Weighted automata model quantitative aspects of systems like the
consumption of resources during executions. Traditionally, the weights
are assumed to form the algebraic structure of a semiring, but
recently also other weight computations like average have been
considered. Here, we investigate quantitative context-free languages
over very general weight structures incorporating all semirings,
average computations, lattices. In our main result, we derive the Chomsky-Sch\"utzenberger Theorem for such quantitative context-free languages, showing that each arises as the image of the intersection of a Dyck language and a recognizable language under a suitable morphism. Moreover, we show that quantitative context-free languages are expressively equivalent to a model of weighted pushdown automata. This generalizes results previously known only for semi\-rings. We also investigate under which conditions quantitative context-free languages assume only finitely many values.
\end{abstract}

\noindent {\bf keywords} weighted context-free grammars, weighted pushdown automata, valuation
  monoids, context-free step functions.

\section{Introduction}	

The Chomsky-Sch\"utzenberger Theorem forms a famous cornerstone in the
theory of context-free languages \cite{chosch63} relating arbitrary
context-free languages to Dyck languages and recognizable languages. A
weighted version of this result reflecting the degrees of ambiguity
was also already given in \cite{chosch63}. For weights taken in
commutative semirings this result was presented in \cite{salsoi78}. For surveys on this we refer the reader to \cite{autberboa97,petsal09}.

Recently, in \cite{chadoyhen08,chadoyhen09} new models of quantitative automata for technical systems have been investigated describing, e.g., the average consumption of resources. In \cite{chavel12} pushdown automata with mean-payoff cost functions were considered which comprize a quantitative modelling of sequential programs with recursion.  These cost functions cannot be computed in semirings. Automata over the general algebraic structure of valuation monoids  were investigated in \cite{dromei10,dromei11}. In valuation monoids, each sequence of weights gets assigned a single weight; examples include products in semirings as well as average computations on reals. Hence automata over valuation monoids include both semiring-weighted automata and quantitative automata. 

It is the goal of this paper to investigate weighted context-free grammars over valuation monoids. Hence we may associate to each derivation, for instance, the average of the costs of the involved productions. We could also associate with each production its  degree of sharpness or truth, as in multi-valued logics, using bounded lattices as valuation monoids. Thereby, we can associate to each word over the underlying alphabet $\Sigma$ such a value (real number, element of a lattice, etc.) indicating its total cost or degree of truth, and any function from $\Sigma^*$  into the value set is called a quantitative language or series. Note that by the usual identification of sets with $\{0,1\}\,$-valued functions, classical languages arise as particular quantitative languages.

\sloppy Now we give a summary of our results. We prove the equivalence of weighted context-free grammars and weighted pushdown automata  over arbitrary valuation monoids (cf. Theorem~\ref{th:PDA=CF}).   In our main result we derive a weighted version of the Chomsky-Sch\"utzenberger Theorem over arbitrary valuation monoids (cf. Theorem~\ref{th:char}). In particular, we show that any  quantitative context-free language arises as the image of the intersection of a  Dyck language and a recognizable language under a suitable weighted morphism, and also as the image of a Dyck language and a recognizable series under a free monoid morphism. Conversely, each quantitative language arising  as such an image is a quantitative context-free  language. This shows that the weighted Chomsky-Sch\"utzenberger  Theorem holds for much more general weighted structures than commutative semirings,  in particular, neither associativity, nor commutativity, nor  distributivity of the multiplication are needed.  In our proofs, due to the lack of the above properties, we cannot use the theory of semiring-weighted automata (cf. \cite{kuisal86,salsoi78}); instead we employ explicit constructions of weighted automata taking care of precise calculations of the weights to deduce our results from the classical, unweighted Chomsky-Sch\"utzenberger Theorem. 
The latter is contained in the weighted result by considering the Boolean semiring $\{0,1\}$.

Finally we consider series which assume only finitely many values.  Such series were important in the context of recognizable series and weighted MSO logic (cf. \cite{berreu88,drogas07}). In contrast to the setting of recognizable series, even over the semiring of natural numbers, characteristic series of inherently ambiguous context-free languages are not context-free (cf. Lemma~\ref{ex:inh-amb}). Here we give sufficient conditions for obtaining context-freeness. Conversely, we show that under suitable local finiteness conditions on the valuation monoid, each quantitative context-free language assumes just finitely many values, each on a context-free language (cf. Theorem~\ref{lm:WPDA-cfstep}). These results seem to be new even for finite semirings. As a consequence, weighted context-free languages over bounded lattices are context-free step functions.

The robustness of automata over valuation monoids is witnessed by the fact  that two other fundamental results of formal language theory, the Kleene and B\"uchi theorems for recognizable languages, were shown to hold not only for semiring-weighted automata \cite{sch61,drogas07}, but also for quantitative automata over valuation monoids \cite{dromei11,mei11,dromei10}.
This raises the question which further results from the theory of semiring-weighted automata could be extended to more general quantitative automata settings including calculations of averages.

This paper combines the papers \cite{drovog13} and \cite{drovog14}
and supplements them by a few additional examples and
more detailed proofs.

\section{Valuation Monoids and Series}
\label{sec:monoids-series}

We define a \emph{unital valuation monoid} to be a tuple $(K,+,\Val,0,1)$ such that 
(i)~$(K,+,0)$ is a commutative monoid,
(ii) $\Val: K^* \rightarrow K$ is a mapping such that $\Val(a) = a$
for each $a \in K$, 
(iii) $\Val(a_1,\ldots, a_n) = 0$ whenever $a_i =0$ for some $1 \le i \le n$, and 
(iv) $\Val(a_1,\ldots, a_{i-1},1,a_{i+1},\ldots, a_n) = \Val(a_1,\ldots, a_{i-1}
,a_{i+1},\ldots, a_n)$ for any $1 \le i \le n$, and (v) $\Val(\varepsilon) = 1$.

Note that, similarly to products where the element 1 is neutral and can be left out, $\Val$ can be considered as a very general product operation in which the unit 1 is neutral as reflected by requirements (iv) and (v). The concept of {\em valuation monoid} was introduced in \cite{dromei10,dromei11} as a structure $(K,+,\Val,0)$ with a mapping  $\Val:  K^+ \rightarrow K$  satisfying requirements (i)-(iii) correspondingly. In \cite{dromei10,dromei11,mei11}, also many examples of valuation monoids were given. For this paper, it will be important that the valuation monoids contain a unit 1. We show that this means no restriction of generality.

\begin{ex}\label{ex:val}\rm

1. \sloppy Let $(K,+,\Val,0)$ be a valuation monoid and let $1 \not\in K$. We put $K' = K \cup \{1\}$ and define $(K',+',\Val',0,1)$ such that $+'$ extends $+$, $x +' 1 = 1 +' x = 1$ for each $x \in K'$, $\Val(\varepsilon) =1$, and $\Val'(a_1,\ldots, a_n) = \Val(b_1,\ldots,b_m)$ where $b_1\ldots b_m$ is the subsequence of $a_1,\ldots, a_n$ excluding $1$'s. Then $(K',+',\Val',0,1)$ is a unital valuation monoid.

2. \sloppy The structure $(\mathbb{R} \cup \{-\infty\}, \sup, \mathrm{avg}, -\infty)$ with $\mathrm{avg}(a_1,\ldots,a_n) = \frac{1}{n} \cdot \sum_{i=1}^n a_i$ is a valuation monoid (with the usual laws for $-\infty$). Applying the procedure of Example 1 to it, we could add $\infty$ as the unit 1, disregarding $\infty$ when 
calculating averages. This leads to a unital valuation monoid $(\mathbb{R} \cup \{-\infty,\infty\}, \sup, \mathrm{avg}, -\infty,\infty)$.

3.  Let $(K,+,\Val,0)$ be a valuation monoid. Note that in Example 1, the unit 1 satisfies $1+'1= 1$. Here we wish to give another extension of $K$ to a unital valuation monoid $K'$ which does not satisfy this law for $1$. Let $K' =  \mathbb{N} \times K$ with componentwise addition $+'$. Given $(m_1,x_1), \ldots, (m_n,x_n)$  we define $\Val'((m_1,x_1), \ldots, (m_n,x_n)) = \Val(x_1, \ldots, x_1, \ldots , x_n\ldots, x_n)$ where this sequence contains $m_i$ copies of $x_i$, if $m_i \not= 0$ and $x_i \not=0$ (for $1 \le i \le n$), otherwise the $x_i$'s are excluded from the sequence to which $\Val$ is applied.
Then $(K',+',\Val',(0,0),(1,0))$ is a unital valuation monoid. 

4. Next we introduce unital valuation monoids where the valuation arises from `local', binary operations as follows.  First, we define a \emph{unital monoid-magma} to be a tuple $(K,+,\cdot,0,1)$ such that $(K,+,0)$ is a commutative monoid, $\cdot:K\times K\rightarrow K$ is a binary operation, and $a \cdot 0 = 0 \cdot a = 0$ and $a \cdot 1 = 1 \cdot a = a$ for every $a \in K$.  Then we can consider each unital monoid-magna $(K,+,\cdot,0,1)$ as the particular unital valuation monoid $(K,+,\Val,0,1)$  where $\Val(a_1,\ldots, a_n) = (\ldots((a_1 \cdot a_2)\cdot a_3)\cdot \ldots )\cdot a_n$. These structures will be important for us in Section \ref{sec:step-functions}.
Now we list some examples of (classes of) unital monoid-magmas which can be viewed as unital valuation monoids in this way:
\begin{itemize} 
\item A unital monoid-magma is a \emph{strong bimonoid}, if the multiplication is associative, and a \emph{semiring}, if the multiplication is associative and distributive (from both sides) over addition. For a range of examples of strong bimonoids which are not semirings we refer the reader to  \cite{drostuvog10}. 

\item The Boolean semiring $\mathbb{B}=(\{0,1\},\lor,\land,0,1)$ allows us to give exact translations between unweighted and $\mathbb{B}$-weighted settings. The semiring $(\mathbb{N},+,\cdot,0,1)$ of natural numbers permits counting.

\item Each bounded lattice $(L,\lor,\land,0,1)$ (i.e. $0\le x\le 1$ for each $x\in L$) is a strong bimonoid. There is a wealth of lattices \cite{bir67,gra03} which are not distributive, hence strong bimonoids but not semirings.

\item 
Let $(S,+,\cdot,0,1)$ be a strong bimonoid and let $n\ge 2$. Then the set of all $(n\times n)$-matrices over $S$ together with the usual pointwise addition and the usual multiplication of matrices forms a unital monoid-magma. We note that this matrix multiplication is not associative if the multiplication of the strong bimonoid is not distributive. We just note that such matrices arise naturally when considering the initial algebra semantics of weighted automata over strong bimonoids, cf. \cite[p. 159]{drostuvog10}.

\item The structure $(\overline{\mathbb{R}},\sup,\mathrm{avg}_2,-\infty,\infty)$ with $\overline{\mathbb{R}}=\mathbb{R}\cup\{-\infty,\infty\}$, and $\mathrm{avg}_2(a,b)=\frac{a+b}{2}$ if $a,b<\infty$, $\mathrm{avg}_2(a,\infty)=\mathrm{avg}_2(\infty,a)=a$ ($a,b\in\overline{\mathbb{R}}$) is a unital monoid-magma, but not a strong bimonoid.

\end{itemize}
\end{ex}

The importance of infinitary sum operations was observed early on in weighted automata theory, cf. \cite{eil74}. In our context, they will arise for ambiguous context-free grammars if a given word has infinitely many derivations.

A monoid $(K,+,0)$ is {\em complete} \cite{eil74} if it has an infinitary
sum operation $\sum_I:K^I \rightarrow K$ for any index set
$I$  such that  
$\sum_{i \in \emptyset} a_i = 0$, $\sum_{i \in \{k\}}
a_i = a_k$,  $\sum_{i \in \{j, k\}} a_i = a_j + a_k$ for $j\not= k$, 
and $\sum_{j\in J} \left( \sum_{i \in I_j} a_i\right) = \sum_{i\in I}
a_i$ if $\bigcup_{j \in J} I_j = I$ and  $I_j \cap I_k = \emptyset$
for $j \not= k$. 

A monoid $(K,+,0)$ is {\em idempotent} if $a+a=a$ for each $a \in K$, and a complete monoid is {\em completely idempotent} if $\sum_{I} a =
a$ for each $a \in K$ and any index set $I$.  

We call a unital valuation monoid $(K,+,\Val,0,1)$ complete,
idempotent, or completely idempotent if $(K,+,0)$ has the respective
property.

\begin{ex}
\rm 1. The Boolean semiring $\mathbb{B}$ 
  and the tropical semiring $(\mathbb{N}\cup \{\infty\}, \min,+,\infty,0)$ 
are complete and completely idempotent. For a wealth of further examples of complete semirings see
\cite[Ch.22]{gol99}.

2. The unital valuation monoid $(\mathbb{R}\cup \{-\infty,\infty\},
\sup,\mathrm{avg},-\infty,\infty)$ (cf. Example \ref{ex:val}(2)) is complete
and completely idempotent.

3. Consider the commutative monoid $(\{0,1,\infty\}, +,0)$  with
$1 +1 = 1$, $1 + \infty = \infty + \infty = \infty$, and $\sum_{I}1 =
\infty$ for any infinite index set $I$ and corresponding natural laws
for infinite sums involving the other elements. This monoid is complete
and idempotent, but not completely idempotent.
\end{ex}

Let $\Sigma$ be an alphabet and $K$ a unital valuation
monoid. A \emph{series} or \emph{quantitative language} over $\Sigma$ and $K$ is a mapping $s:
\Sigma^* \rightarrow K$. As usual, we denote $s(w)$ by $(s,w)$.  
The \emph{support} of $s$ is the set $\supp(s) = \{w \in \Sigma^* \mid (s,w) \not= 0\}$.  The class of all series over $\Sigma$ and $K$ is denoted by $K\fl \Sigma^*\fr$.

Let $s,s' \in K\fl \Sigma^*\fr$ be series.  We define the sum $s+s'$  by letting $(s+ s',w) = (s,w) + (s',w)$ for each $w \in \Sigma^*$. 
A family of series $(s_i \mid i \in I)$ is \emph{locally finite} if
for each $w \in \Sigma^*$ the set $I_w = \{i \in I \mid (s_i,w) \not=
0\}$ is finite. In this case or if $K$ is complete, we define $\sum_{i \in I} s_i \in K\fl \Sigma^*\fr$ by letting $(\sum_{i \in I} s_i, w) = \sum_{i \in I_w} (s_i,w)$ for every $w \in \Sigma^*$. 
For $L \subseteq \Sigma^*$, we define the {\it characteristic series}
$\one_{L} \in K\fl \Sigma^*\fr$ by $(\one_L,w) = 1$ if $w \in L$, and
$(\one_L,w) = 0$ otherwise  for $w \in \Sigma^*$.

\begin{quote}
{\em In the rest of this paper, let $(K,+,\Val,0,1)$ denote an arbitrary unital valuation monoid, unless specified otherwise.
}
\end{quote}

\section{Weighted Context-Free Grammars}

In this section, we introduce our notion of weighted context-free grammars and we present basic properties. 
A \emph{context-free grammar} (CFG) is a tuple
$\G = (N,\Sigma,Z,P)$ where $N$ is a finite set (\emph{nonterminals}),
$\Sigma$ is an alphabet with $N\cap \Sigma =\emptyset$ (\emph{terminals}),
$Z \in N$ (\emph{initial nonterminal}), and $P \subseteq N \times
(N \cup \Sigma)^*$ is a finite set (\emph{productions}).

For every production $\rho = (A \rightarrow \xi) \in P$ we define the
binary relation $\stackrel{\rho}{\Rightarrow}$ on  $(N \cup \Sigma)^*$ such that for every $w \in \Sigma^*$ and $\zeta \in (N \cup \Sigma)^*$, we have $w\, A \, \zeta \stackrel{\rho}{\Rightarrow} w \, \xi \, \zeta$.
A \emph{(leftmost) derivation} of $\G$ is a sequence $d = \rho_1\ldots\rho_n$ of
productions $\rho_i \in P$ such that there are sentential forms
$\xi_0,\ldots,\xi_n$ with $\xi_{i-1} \stackrel{\rho_i}{\Rightarrow}
\xi_i$ for every $1 \le i \le n$. We abbreviate this derivation by $\xi_0  \stackrel{d}{\Rightarrow} \xi_n$.
Let $A\in N$ and  $w \in \Sigma^*$. An \emph{$A$-derivation
  of $w$} is a derivation $d$ such that $A \stackrel{d}{\Rightarrow}
w$. We let $D(A,w)$ denote the set of all $A$-derivations
of $w$. And we let $D(w)$ denote the set  $D(Z,w)$ of all {\em derivations of $w$}. The \emph{language generated by $\G$} is the set 
$L(\G) = \{w \in \Sigma^* \mid D(w) \not= \emptyset\}$. 

We say that $\G$ is {\em ambiguous} if there is a $w \in L(\G)$ such that
$|D(w)| \ge 2$; otherwise $\G$ is {\em unambiguous}. A context-free
language $L$ is {\em inherently ambiguous} if every context-freeG $\G$ with $L =
L(\G)$ is ambiguous.

Next let $K$ be a unital valuation monoid. A \emph{context-free grammar with weights in $K$} is a tuple $\G = (N,\Sigma,Z,P,\wt)$ where $(N,\Sigma,Z,P)$ is a CFG and $\wt\colon P \rightarrow K$ is a mapping (\emph{weight assignment}). We say that ${\cal G}$ is unambiguous if the underlying CFG is unambiguous.

The \emph{weight} of a derivation $d = \rho_1\ldots \rho_n$ is the element in $K$ defined by
\[
\wt(d) = \Val( \wt(\rho_1), \ldots, \wt(\rho_n))\enspace.\]

We say that $\G$ is a \emph{weighted context-free grammar} (WCFG) if
(i) $\{ d \in D(w) \mid \wt(d) \not=0\}$ is  finite for every $w \in \Sigma^*$ or (ii) $K$ is complete. In this case we define the  \emph{quantitative language} of $\G$ to be the series $\beh \G \beh \in K\fl \Sigma^*\fr$ given for every $w \in \Sigma^*$ by  
\[
(\beh \G \beh,w)  = \sum_{d \in D(w)} \wt(d)\enspace.
\]
Note that this sum exists by our assumptions on a WCFG.
A series $s  \in K\fl \Sigma^*\fr$ is a \emph{quantitative
  context-free language} if there is a WCFG $\G$ such that $s = \beh
\G \beh$. The class of all quantitative context-free languages over
$\Sigma$ and  $K$ is denoted by $\CF(\Sigma,K)$. Moreover, we let
$\uCF(\Sigma,K)$ comprise all series $\beh \G \beh$ where $\G$ is an
unambiguous WCFG. We say that two WCFG are \emph{equivalent}, if they have the same quantitative language.

Clearly, any CFG $\G$ can be transformed into a WCFG over the Boolean
semiring $\mathbb{B}$ by adding the weight assignment $\wt\colon P
\rightarrow \mathbb{B}$ such that $\wt(\rho) = 1$ for each $\rho
\in P$. Then for each $w \in \Sigma^*$ we have $w \in L(\G)$ if and
only if $(\beh \G \beh,w) = 1$, i.e., $\beh \G\beh =
\one_{L(\G)}$.  Consequently, a language $L$ is context-free if and
only if $\one_L \in \CF(\Sigma,\mathbb{B})$. 
 This shows that WCFG form a generalization of CFG.

We say that a CFG $\G=(N,\Sigma,Z,P,\wt)$ with weights is
{\em proper} if the right-hand side of each rule is an element of $(N
\cup \Sigma)^+ \setminus N$ (cf. \cite[p.302]{kuisal86}), i.e., $\G$ contains neither chain 
productions $A\rightarrow B$ nor $\varepsilon$-productions $A\rightarrow \varepsilon$. Then obviously the set  $D(w)$ is finite for every $w \in
\Sigma^*$. 

\begin{ob} A proper CFG with weights is a WCFG.
\end{ob}

A WCFG $\G$ is in \emph{head normal form} if every production has the form $A \rightarrow xB_1\ldots B_k$ 
where $x \in \Sigma \cup \{\varepsilon\}$, $k \ge 0$,  and
$A,B_1,\ldots,B_k \in N$. By a standard construction we now obtain the following.

\begin{lm}\label{lm:WCFG:normal-form} For every (unambiguous) WCFG there is an
  equivalent (unambiguous) WCFG in head normal form.
\end{lm}

\begin{proof} Let $\G = (N,\Sigma,Z,P,\wt)$ be a WCFG. We construct the CFG with weights $\G' = (N',\Sigma,Z,P',\wt')$ such that 
\begin{itemize}
\item $N' = N \cup \{A_\sigma \mid \sigma \in \Sigma\}$,
\item $P'$ and $\wt'$ are determined as follows. 
\begin{itemize}
\item If $A \rightarrow \varepsilon$ is in $P$, then $A \rightarrow \varepsilon$ is in $P'$; moreover, $\wt'(A \rightarrow \varepsilon) = \wt(A \rightarrow \varepsilon)$.

\item If $A \rightarrow X\xi$ in $P$ with $X \in N\cup \Sigma$ and $\xi \in (N\cup \Sigma)^*$, then $A \rightarrow X\xi'$ is  in $P'$ where $\xi'$ is obtained from $\xi$ by replacing every $\sigma \in \Sigma$ by $A_\sigma$; moreover, $\wt'(A \rightarrow X\xi') = \wt(A \rightarrow X\xi)$.

\item For every $\sigma \in \Sigma$, the production $A_\sigma \rightarrow \sigma$ is in $P'$; moreover, $\wt'(A_\sigma \rightarrow \sigma) = 1$.   
\end{itemize}
\end{itemize}
Then, for every $w \in \Sigma^*$, we have that  every derivation $d$ of $w$ by $\G$ corresponds naturally to a uniquely determined derivation of $w$ by  $\G'$, and vice versa.  Thus $\G'$ is also a WCFG and corresponding successful derivations have the same weight. Thus $\beh \G \beh = \beh \G'\beh$.  Clearly, $\G$ is unambiguous if and only if $\G'$ is unambiguous.
\end{proof}

\begin{ex}\rm \label{ex:wcfg} We consider the set of all arithmetic expressions over   addition, multiplication, and the variable $x$. Assuming that the calculation  of the addition (and multiplication) of two values needs $n \in \mathbb{N}$ (resp., $m \in \mathbb{N}$) machine  clock cycles, we
  might wish to know the average number of clock cycles the machine  needs to calculate any of the operations occurring in an expression. 

For this we consider the unital valuation monoid $(\mathbb{R} \cup \{-\infty, \infty\}, \sup, \mathrm{avg},$ $-\infty, \infty)$ as above and the WCFG $\G = (N,\Sigma,E,P,\wt)$ with the productions
\begin{center}
$\rho_1:  E \rightarrow  (E + E)$,
$\rho_2: E  \rightarrow  (E * E)$, 
$\rho_3: E \rightarrow  x$
\end{center}
and $\wt(\rho_1) = n$, $\wt(\rho_2) = m$, $\wt(\rho_3) = 1$.
For the expression $w = ((x * x) + (x* x))$, we have that $D(w) = \{d\}$ with $d= \rho_1 d' d'$ and  $d' = \rho_2 \rho_3 \rho_3$. In fact, $\G$ is unambiguous.  Then
$(\beh {\cal G} \beh, w) = \wt(d) = \Val(n, m,1,1, m,1,1)$ which is equal to $\mathrm{avg}(n,m,m) = \frac{n + 2\cdot m}{3}$.
\end{ex}


\section{Weighted Pushdown Automata}

In this section, we introduce our notion of weighted pushdown automata, and we derive a few basic properties. First let us fix our notation for pushdown automata. A \emph{pushdown automaton} (PDA) over $\Sigma$ is a tuple $\M = (Q,\Sigma,\Gamma,q_0,\gamma_0,F,T)$ where
$Q$ is a finite set (\emph{states}),
$\Sigma$ is an alphabet (\emph{input symbols}),
 $\Gamma$ is an alphabet (\emph{pushdown symbols}),
$q_0 \in Q$ (\emph{initial state}),
$\gamma_0 \in \Gamma$ (\emph{initial pushdown symbol}),
 $F \subseteq Q$ (\emph{final states}), and
$T \subseteq Q \times \big(\Sigma \cup \{\varepsilon\}\big) \times
\Gamma \times Q \times \Gamma^*$ is a finite set (\emph{transitions}).
For a transition $(q,x,\gamma,p,\pi)$, we call $q$, $x$, and $p$ its \emph{source state}, \emph{label}, and \emph{target state}, respectively.

For every transition $\tau = (q,x,\gamma,p,\pi) \in T$ we define 
the binary relation $\vdash^\tau$ on $Q \times \Sigma^*
\times \Gamma^*$ such that for every $w \in \Sigma^*$ and $\mu \in \Gamma^*$, we have 
$(q,xw,\gamma\mu) \vdash^\tau (p,w,\pi\mu)$.
A \emph{computation} is a sequence $\theta = \tau_1 \ldots \tau_n$ of transitions $\tau_i$ such that there are configurations $c_0,\ldots,c_n$ with  $c_{i-1} \vdash^{\tau_i} c_{i}$ for every $1 \le i \le n$. We abbreviate this computation by $c_0 \vdash^\theta c_n$.
The \emph{label} of a computation  $\tau_1 \ldots \tau_n$ is the
sequence of labels of the involved transitions. Let $w \in \Sigma^*$
and $q \in Q$. A \emph{$q$-computation on $w$} is a
computation $\theta$ such that $(q,w,\gamma_0)
\vdash^\theta (p,\varepsilon,\varepsilon)$ for some $p \in
F$. We let $\Theta(q,w)$ denote the set of all $q$-computations on $w$, and we let $\Theta(w) =\Theta(q_0,w)$. 
The \emph{language recognized by $\M$} is the set $L(\M) = \{ w \in \Sigma^* \mid \Theta(w) \not= \emptyset\}$. That means, we consider acceptance of words by final state and empty pushdown.

\begin{ob}\label{obs:decomp-PDA} 
Let $(q,v,\gamma_1 \ldots \gamma_k) \stackrel{\eta}{\vdash} (p,\varepsilon,\varepsilon)$ be a computation. Put $p_0 = q$. Then for $i=1,\ldots, k$ we obtain successively a uniquely determined shortest computation $\eta_i$, state $p_i\in Q$, and word $v_i\in\Sigma^*$ such that $(p_{i-1},v_i,\gamma_i) \stackrel{\eta_i}{\vdash}(p_i,\varepsilon,\varepsilon)$ and $\eta = \eta_1 \ldots \eta_k$,
 $v = v_1 \ldots v_k$,
 and $p_k = p$.
\end{ob}

Let $\M$ be any PDA. We say that $\M$ is {\em ambiguous} if there is a $w \in L(\M)$ such that $|\Theta(w)| \ge 2$; otherwise $\M$ is {\em unambiguous}.

Next let $K$ be a unital valuation monoid. A \emph{pushdown automaton with weights in $K$} is a tuple $\M = (Q,\Sigma,\Gamma,q_0,\gamma_0,F,T,\wt)$ where
$(Q,\Sigma,\Gamma,q_0,\gamma_0,F,T)$ is a PDA and  $\wt\colon T \rightarrow K$ is a mapping (\emph{weight assignment}). We say that ${\cal M}$ is unambiguous if the underlying PDA is unambiguous.

The \emph{weight} of a computation $\theta = \tau_1 \ldots \tau_n$ is the element in $K$ defined by
\[
\wt(\theta) = \Val( \wt(\tau_1), \ldots, \wt(\tau_n))\enspace.
\]

We say that $\M$ is a \emph{weighted pushdown automaton} (WPDA) if
(i) $\{\theta \in \Theta(w) \mid \wt(\theta) \not=0\}$ is finite for every $w \in \Sigma^*$ or (ii) $K$ is complete. 
 In this case we define the \emph{quantitative behavior} of $\M$ to be the series $\beh \M \beh \in K\fl \Sigma^* \fr$ given for every $w \in \Sigma^*$ by  
\[
(\beh \M\beh, w) = \sum_{\theta \in \Theta(w)} \wt(\theta)\enspace.
\] 
The class of  quantitative behaviors of all WPDA over $\Sigma$ and  $K$ is denoted by $\PDA(\Sigma,K)$.  Moreover, we let
$\uPDA(\Sigma,K)$ comprise all series $\beh \M \beh$ where $\M$ is an
unambiguous WPDA. We say that two WPDA are \emph{equivalent} if they have the same quantitative behavior.

Clearly, any PDA $\M$ can be transformed into a WPDA over the Boolean
semi\-ring $\mathbb{B}$ by adding the weight assignment $\wt\colon T
\rightarrow \mathbb{B}$ such that $\wt(\tau) = 1$ for each $\tau
\in T$. Then for each $w \in \Sigma^*$ we have $w \in L(\M)$ if and
only if $(\beh \M \beh,w) = 1$, i.e., $\beh \M\beh =
\one_{L(\M)}$.  Consequently, a language $L$ is recognized by a PDA if and
only if $\one_L \in \PDA(\Sigma,\mathbb{B})$. This shows that WPDA form a generalization of PDA.

We say that a PDA $\M
=(Q,\Gamma,q_0,\gamma_0,F,T,\wt)$ with weights is {\em proper} if
$(q,\varepsilon,\gamma,p,\pi) \in T$ implies $|\pi| \ge 2$ (cf. \cite[p.172]{kuisal86}), i.e., $\M$ extends its pushdown in each
$\varepsilon$-transition. Then obviously the set  $\Theta(w)$ is finite for every $w
\in \Sigma^*$.

\begin{ob} Each proper PDA with weights is a WPDA.
\end{ob}

A WPDA  $\M = (Q,\Sigma,\Gamma,q_0,\gamma_0,F,T,\wt)$ is \emph{state
  normalized} if
\begin{itemize}
\item there is no transition in $T$ with $q_0$ as target state, 
\item  $F$ is a singleton, say, $F = \{q_f\}$, and 
\item there is no transition in $T$ with $q_f$ as source state.
\end{itemize}
By a standard construction we obtain the following.

\begin{lm}\label{lm:WPDA-normalized} For every (unambiguous) WPDA there is an
  equivalent state normalized (unambiguous)  WPDA. 
\end{lm}
\begin{proof} \sloppy Let $\M = (Q,\Gamma,q_0,\gamma_0,F,T,\wt)$. We construct
  the PDA with weights $\M' = (Q',\Gamma',q_0',\gamma_0',\{q_f\},T',\wt')$ with 
\begin{itemize}
\item $Q' = Q \cup \{q_0',q_f\}$ with $Q \cap \{q_0',q_f\} = \emptyset$, 
\item $\Gamma' = \Gamma \cup \{\gamma_0'\}$,
\item $T' = T \cup
  \{\tau_{\mathrm{in}}\} \cup
  \{\tau_p \mid p \in F \}$,  $\tau_{\mathrm{in}} = (q_0',\varepsilon,\gamma_0',q_0,\gamma_0\gamma_0')$ and $\tau_p = (p,\varepsilon,\gamma_0',q_f,\varepsilon)$
  and 
\item $\wt'|_T = \wt$ and
  $\wt'(\tau_{\mathrm{in}}) =
  \wt(\tau_p) = 1$, for each $p\in F$. 
\end{itemize}
\sloppy Let $\theta = \tau_1 \tau_2 \ldots \tau_n$ be a computation of $\M$ on $w$ and let $p$ be the target state of $\tau_n$. Then $\theta' =
\tau_{\mathrm{in}} \tau_1 \tau_2 \ldots \tau_n \tau_p$ is a computation of $\M'$ on $w$ and $\wt'(\theta') = \wt(\theta)$. Vice versa, every computation of $\M'$ on $w$ has the form $\tau_{\mathrm{in}} \tau_1 \tau_2 \ldots \tau_n \tau_p$ for some $p \in F$. Then $\tau_1 \tau_2 \ldots \tau_n$ is a computation of $\M$ on $w$.

Moreover,  $\M'$ is a WPDA. Thus we have $\beh \M\beh = \beh \M'\beh$. Clearly,  $\M$ is unambiguous if and only if $\M'$ is unambiguous.
\end{proof}

Next we show that WPDA with just one state are as powerful as
arbitrary WPA using the classical triple construction (cf. e.g. \cite[Lecture 25]{koz97}).

\begin{lm}\label{lm:WPDA-one-state} For every (unambiguous) WPDA there is an equivalent (unambiguous) WPDA with just one state.
\end{lm}
\begin{proof}\sloppy Let $\M$ be a WPDA. By Lemma
  \ref{lm:WPDA-normalized} we can assume that $\mathcal{M}$ is state normalized and has the form $(Q,\Gamma,q_0,\gamma_0,\{q_f\},T,\wt)$. Then we construct the PDA with weights $\M' = (\{\ast\},\Gamma',\ast, (q_0,\gamma_0,q_f), \{\ast\},T',\wt')$ and $\Gamma' = Q \times \Gamma \times Q$ as follows. For each transition 
\[(q,x,\gamma,p_0,\gamma_1\ldots \gamma_k) \in T
\] 
and every $p_1,\ldots,p_k \in Q$, the transition 
\[
(\ast,x,(q,\gamma,p_k),\ast,(p_0,\gamma_1,p_1)(p_1,\gamma_2,p_2)\ldots (p_{k-1},\gamma_k,p_k))
\]
is in $T'$. For $k=0$ this reads: if $(q,x,\gamma,p_0,\varepsilon) \in T$, then $(\ast,x,(q,\gamma,p_0),\ast,\varepsilon) \in T'$.

Moreover, 
\begin{align*}
&\wt'((\ast,x,(q,\gamma,p_k),\ast,(p_0,\gamma_1,p_1)(p_1,\gamma_2,p_2)\ldots (p_{k-1},\gamma_k,p_k))) \\
&= \wt((q,x,\gamma,p_0,\gamma_1\ldots \gamma_k)))\enspace.
\end{align*}

For every computation $\theta$ of ${\cal M}$ we construct a
computation  $\varphi(\theta)$ of $\M'$ as follows. Let $\theta \in
\Theta_\M(w)$ and $\tau = (q,x,\gamma,p,\gamma_1\ldots \gamma_k)$ be a transition which occurs in $\theta$. Then
$\theta$ can be decomposed into
\[
(q_0,w,\gamma_0) 
\stackrel{\theta_1}{\vdash} (q,xu,\gamma\mu)
\stackrel{\tau}{\vdash} (p,u,\gamma_1\ldots \gamma_k\mu) 
\stackrel{\theta_2}{\vdash} (q_f,\varepsilon,\varepsilon) \enspace.
\]
Let $\eta$, $v$, and $p'$ be the, resp., shortest prefix of
$\theta_2$, shortest prefix of $u$, and uniquely determined state such that 
\[
(p,v,\gamma_1\ldots \gamma_k) \stackrel{\eta}{\vdash} (p',\varepsilon,\varepsilon)\enspace.
\]
Put $p_0=p$. Then, by Oberservation \ref{obs:decomp-PDA}, for $i=1,\ldots,k$ we can successively find a shortest word $v_i \in \Sigma^*$, a uniquely determined state $p_i \in Q$, and a shortest computation $\eta_i$ such that $(p_{i-1},v_i,\gamma_i) \stackrel{\eta_i}{\vdash} (p_i,\varepsilon,\varepsilon)$ and  $v = v_1 \ldots v_k$, $p_k=p'$, and $\eta = \eta_1\ldots \eta_k$.
Then we replace the occurrence $\tau$ in $\theta = \theta_1 \tau \theta_2$ by 
\[
\varphi(\tau) = (\ast,x,(q,\gamma,p_k),\ast,(p,\gamma_1,p_1)(p_1,\gamma_2,p_2)\ldots (p_{k-1},\gamma_k,p'))\enspace.
\]
We extend $\varphi$ to computations by letting 
\[
\varphi(\tau_1\ldots \tau_n) = \varphi(\tau_1) \ldots \varphi(\tau_n)\enspace.
\]
In particular, $\varphi: \Theta_\M(w) \rightarrow \Theta_{\M'}(w)$ is
bijective. Thus, since $\M$ is a WPDA, also $\M'$ is a WPDA. 
Since $\wt(\tau) = \wt(\varphi(\tau))$, we have $\wt(\theta) =
\wt(\varphi(\theta))$. Hence, $\beh \M \beh = \beh \M'\beh$.

It is clear that $\M$ is unambiguous if and only if $\M'$ is unambiguous.
\end{proof}

A classical construction of the union of two state normalized WPDA
shows that $\PDA(\Sigma,K)$ is closed under sums.

\begin{lm}\label{lm:WPDA-sum} Let $s_1,s_2 \in \PDA(\Sigma,K)$. Then $s_1 + s_2 \in \PDA(\Sigma,K)$.
\end{lm}
\begin{proof} Let $\M_1 = (Q_1,\Gamma_1,q_{0,1},\gamma_{0,1},F_1,T_1,\wt_1)$ and 
$\M_2 = (Q_2,\Gamma_2,q_{0,2},\gamma_{0,2},F_2,T_2,\wt_2)$ be state normalized WPDA (compare Lemma \ref{lm:WPDA-normalized})  such that $s_1 = \beh \M_1\beh$ and $s_2 = \beh \M_2\beh$. By renaming, we may assume that 
\begin{itemize}
\item $q_{0,1}= q_{0,2}$ and $Q_1 \cap Q_2 = \{q_{0,1}\}$; henceforth we will denote this initial state by $q_0$, and 
\item  $\gamma_{0,1} = \gamma_{0,2}$; henceforth we will denote this initial pushdown symbol by $\gamma_0$.
\end{itemize}

Now we construct the WPDA $\M = (Q,\Gamma,q_0,\gamma_0,F,T,\wt)$ with
$Q = Q_1 \cup Q_2$, $\Gamma = \Gamma_1 \cup \Gamma_2$, $F = F_1 \cup
F_2$, and $T = T_1 \cup T_2$. Moreover, we let $\wt|_{T_1} = \wt_1$
and  $\wt|_{T_2} = \wt_2$. Then $\beh \M \beh = \beh \M_1\beh + \beh \M_2\beh$.
\end{proof}

We mention that in \cite{chavel12} pushdown games with quantitative objectives were investigated. 
Such games are formalized on the basis of paths through pushdown systems where the latter are particular pushdown automata with weights: the input alphabet $\Sigma$ is a singleton and no $\varepsilon$-transition occurs. 
Moreover, as weight structure, pushdown systems employ the set of integers with mean-payoff. 
Roughly, the mean-payoff of a computation is the average of its transition weights (taking the limit superior of the averages of finite prefixes on infinite computations). 
Then in \cite{chavel12} game-theoretic problems on the set of all paths for which the mean-payoff is above a given threshold are investigated.
Finally, we note that weighted pushdown systems over bounded idempotent semirings were used in interprocedural dataflow analysis \cite{repschjhamel05}.


\section{Equivalence of WCFG and WPDA}

A classical result says that a language $L$ is context-free iff  $L$  is accepted by a pushdown automaton. This was extended to algebraic series and weighted pushdown automata with weights taken in semirings in \cite{kuisal86}. The goal of this small section is to prove the generalization to arbitrary unital valuation monoids. 

For this we use the following concept. 
Let $\M = (\{\ast\},\Sigma,\Gamma,\ast,\gamma_0,\{\ast\},T,\wt_{\M})$ be a WPDA over $K$ with one state and $\G = (N,\Sigma,Z,P,\wt_{\G})$ be a WCFG over $K$ in head normal form. We say that $\M$ and $\G$ are \emph{related} if
$\Gamma = N$,
$\gamma_0 = Z$, 
$\tau = (\ast,x,A,\ast,B_1B_2\ldots B_n) \in T$ iff 
$\rho = (A \rightarrow xB_1B_2\ldots B_n)$ is in $P$;
 $\wt_{\M}(\tau) = \wt_{\G}(\rho)$ if $\tau$ and $\rho$ correspond to each other as above. Then the following lemma is easy to see (cf. e.g. \cite{koz97}).

\begin{lm}\label{lm:related} Let $\M$ be a WPDA with one state  and $\G$ be a WCFG in head normal form. If $\M$ and $\G$ are related, then $\beh \M \beh = \beh \G\beh$. Moreover, $\M$ is unambiguous iff $\G$ is unambiguous.
\end{lm}

\begin{proof} For every $w \in \Sigma^*$, the set $\Theta(w)$ of computations on $w$ corresponds bijectively to the set $D(w)$ of derivations of $w$, and this correspondence preserves their weights. This implies $(\beh \M \beh,w) = (\beh \G\beh,w)$.
\end{proof}

The previous lemma and the normal forms of WPDA and WCFG imply the following theorem.

\begin{theo}\label{th:PDA=CF} For every alphabet $\Sigma$ and
  unital valuation monoid $K$ we have $\PDA(\Sigma,K) =
  \CF(\Sigma,K)$ and $\uPDA(\Sigma,K) =
  \uCF(\Sigma,K)$. 
\end{theo}

\begin{proof} If $s \in \PDA(\Sigma,K)$, then 
by Lemma \ref{lm:WPDA-one-state} there is a WPDA ${\cal M}$ with one state such that $\beh {\cal M}\beh =s$. Next we construct the WCFG ${\cal G}$ which is related to ${\cal M}$. By Lemma \ref{lm:related} we obtain that $s \in \CF(\Sigma,K)$. Correspondingly, the converse follows from Lemmas \ref{lm:WCFG:normal-form} and \ref{lm:related}. 
\end{proof}

\section{Theorem of Chomsky-Sch\"utzenberger}
\label{sec:Chomsky-Sch}

In this section let $K$ again be a unital valuation monoid. The
goal of this section will be to prove a quantitative version of the
Chomsky-Sch\"utzenberger Theorem. 
Recently,  in \cite{hul09} the Chomsky-Sch\"utzenberger Theorem has been used as a pattern for a parsing algorithm of probabilistic context-free languages.

Let $Y$ be an alphabet. Then we let $\overline{Y} = \{\overline{y}
\mid y \in Y\}$. The \emph{Dyck language over $Y$}, denoted by $\D_Y$,
is the language which is generated by the CFG $\G_Y =
(N,Y \cup \overline{Y}, Z,P)$ with $N = \{Z\}$ and the rules $Z
\rightarrow y Z \overline{y}$ for any $y \in Y$, $Z \rightarrow ZZ$, and $Z \rightarrow
\varepsilon$. 

Next we introduce monomes and alphabetic morphisms.
A series $s \in K\fl \Sigma^*\fr$  is called a \emph{monome} if $\supp(s)$ is empty or a singleton. 
If $\supp(s) = \{w\}$, then we also write $s = (s,w).w$\enspace.
We let $K[\Sigma \cup \{\varepsilon\}]$ denote the set of all monomes with support in $\Sigma \cup \{\varepsilon\}$.
 
Let $\Delta$ be an alphabet and $h: \Delta \rightarrow K[\Sigma
\cup \{\varepsilon\}]$ be a mapping. The  \emph{alphabetic morphism induced by $h$} is the mapping $h': \Delta^* \rightarrow  K\fl \Sigma^*\fr$ such that for every $n \ge 0$, $\delta_1,\ldots,\delta_n \in \Delta$ with $h(\delta_i) = a_i . y_i$ we have 
$
h'(\delta_1 \ldots \delta_n) = \Val(a_1,\ldots,a_n).y_1\ldots y_n\enspace.
$
Note that $h'(v)$ is a monome for every $v \in \Delta^*$, and $h'(\varepsilon) = 1.\varepsilon$.
If $L \subseteq \Delta^*$ such that the family $(h'(v) \mid v \in L)$
is locally finite or if $K$ is complete, we let $h'(L) = \sum_{v\in L}
h'(v)$. In the sequel we identify $h'$ and $h$.

We also call a mapping $h: \Delta \rightarrow \Sigma \cup \{\varepsilon\}$ and its unique extension  to a morphism from $\Delta^*$ to $\Sigma^*$ an \emph{alphabetic morphism}. In this case, if $r \in K\fl \Delta^*\fr$ is such that $\{v \in h^{-1}(w) \mid (r,v) \not=0\}$ is finite for each $w \in \Sigma^*$, or if $K$ is complete, we define $h(r) \in K\fl \Sigma^*\fr$  by letting $(h(r),w) = \sum_{v \in \Delta^*, h(v)=w} (r,v)$.

Next we introduce the intersection of a series with a language as
follows. Let $s \in K\fl \Sigma^*\fr$ and $L \subseteq \Sigma^*$. We define the
series $s \cap L \in K\fl \Sigma^*\fr$ by letting $(s \cap L, w) = (s,w)$ if $w \in L$, and $(s \cap L, w) = 0$ otherwise.

Finally, a {\it weighted finite automaton} over $K$ and $\Sigma$ (for short: WFA) is a tuple ${\cal A} = (Q,q_0,F,T,\wt)$ where $Q$ is a finite set (states), $q_0 \in Q$ (initial state), $F \subseteq Q$ (final states), $T \subseteq Q \times \Sigma \times Q$ (transitions), and $\wt\colon T \rightarrow K$ (transition weight function). 
We call ${\cal A}$ {\em deterministic} if for every $q \in Q$ and $\sigma \in \Sigma$, there is at most one $p \in Q$ with $(q,\sigma,p) \in T$. 

If $w = \sigma_1\ldots \sigma_n \in \Sigma^*$ where $n \ge 0$ and $\sigma_i \in \Sigma$, a {\em
   path $P$  over $w$} is a sequence  $P = (q_0,\sigma_1,q_1) \ldots
 (q_{n-1},\sigma_n, q_n) \in T^*$. The path $P$ is {\em successful} if
 $q_n \in F$.  The {\it weight of} $P$ is the value
$
\wt(P) = \mathrm{val}(\wt((q_0,\sigma_1,q_1)), \ldots ,\wt((q_{n-1},\sigma_n, q_n))) \enspace.
$
The {\it behavior} of ${\cal A}$ is the series $\beh {\cal A} \beh \in K\fl \Sigma^* \fr$ such that for every $w \in \Sigma^*$, 
$(\beh {\cal A} \beh , w) = \sum_{\substack{P \text{ succ. path}\\\text{over } w}}\wt(P).$
A series $s \in K\fl \Sigma^*\fr$ is called {\em deterministically recognizable} if $s = \beh {\cal A}\beh$ for some deterministic WFA ${\cal A}$.

Our main result will be:

\begin{theo} \label{th:char} Let $K$ be a unital valuation monoid and  $s \in K\fl \Sigma^*\fr$ be a series. Then the following four statements are equivalent.
\begin{enumerate}
\item $s \in \CF(\Sigma,K)$.
\item There are
  an alphabet $Y$, a recognizable language $R$ over
  $Y\cup\overline{Y}$, and an alphabetic morphism $h: Y \cup \overline{Y}
  \rightarrow K[\Sigma \cup \{\varepsilon\}]$ such that $s = h(\D_Y \cap R)$.
\item There are an alphabet $\Delta$, an unambiguous CFG $\G$ over
  $\Delta$, and an alphabetic morphism $h: \Delta \rightarrow
  K[\Sigma \cup \{\varepsilon\}]$ such that $s = h(L(\G))$.

\item  There are an alphabet $Y$, a deterministically recognizable series $r \in K\fl (Y \cup \bar{Y})^*\fr$, and an alphabetic morphism $h: Y \cup \overline{Y}
  \rightarrow \Sigma \cup \{\varepsilon\}$ such that $s = h(r \cap \D_Y)$.
\end{enumerate}
Moreover, if $K$ is complete and completely idempotent, then 1-4
are also
equivalent to:
\begin{enumerate}
\item[5.] There are an alphabet $\Delta$, a context-free language $L$
  over $\Delta$, and an alphabetic morphism $h: \Delta \rightarrow
  K[\Sigma \cup \{\varepsilon\}]$ such that $s = h(L)$.
\end{enumerate}
\end{theo}

The following lemma proves the implication $1 \Rightarrow 3$ of Theorem
\ref{th:char}.

\begin{lm}\label{th:mezwri}  Let $s \in \CF(\Sigma,K)$. Then there are
  an alphabet $\Delta$, an unambiguous CFG $\G$ over $\Delta$, and an
  alphabetic morphism $h: \Delta \rightarrow K[\Sigma \cup
  \{\varepsilon\}]$ such that 
$s = h(L(\G))$.
\end{lm}
\begin{proof} By Lemma \ref{lm:WCFG:normal-form} we can assume that $s = \beh \HH \beh$ for some WCFG $\HH = (N,\Sigma,Z,P,\wt)$ in head normal form. 

We let $\Delta = P$, and we construct the CFG $\G = (N,P,Z,P')$ and the mapping $h:
P \rightarrow K\langle\Sigma \cup \{\varepsilon\}\rangle$ such
that, if $\rho = (A \rightarrow x B_1\ldots B_k)$ is in $P$, then $A
\rightarrow \rho B_1 \ldots B_k$ is in $P'$ and we define $h(\rho) =
\wt(\rho).x$. Obviously, $\G$ is unambiguous.

By definition of $h$, we have that $h(d) =
\Val(\wt(\rho_1),\ldots,\wt(\rho_n)).w$ for every $w \in \Sigma^*$ and $d = \rho_1\ldots
\rho_n\in D_{\cal H}(w)$. Hence $\wt(d) = (h(d),w)$. 

It is clear that $L(\G) = \bigcup_{w \in \Sigma^*} D_{\cal H}(w)$ and $D_{\cal H}(w) \cap D_{\cal H}(w') = \emptyset$ for every $w \not= w'$. 
 Hence, $\{d \in L(\G) \mid (h(d),w) \not= 0\} \subseteq D_{\cal H}(w)$,
which is finite by definition of WCFG in case $K$ is not complete. Thus, $(h(d) \mid d
\in L(\G))$ is locally finite if  $K$ is not complete.

Then for every $w \in \Sigma^*$ we have 
\[
\begin{array}{l}
(\beh \HH\beh, w) = \sum_{d \in D_{\HH}(w)} \wt(d) = \sum_{d \in D_{\HH}(w)}
(h(d),w)\\
 =^*  \sum_{d \in L(\G)} (h(d),w) = \left(  \sum_{d \in L(\G)}
  h(d) , \; w \right) = (h(L(\G)),w)
\end{array} 
\] 
where the equation $*$ holds, because for every $d \in L({\cal G}) \setminus D_{\cal H}(w)$ we have $(h(d),w) = 0$. Thus $s = h(L(\G))$. 
\end{proof}

The following lemma proves the implication $3 \Rightarrow 1$ of Theorem
\ref{th:char}.

\begin{lm}\label{lm:closure} Let $L$ be a context-free language over
  $\Delta$ and $h: \Delta \rightarrow K[\Sigma \cup
  \{\varepsilon\}]$ an alphabetic morphism such that $(h(v) \mid
  v \in L)$ is locally finite in case $K$ is not complete. If $L$ can be generated by some
  unambiguous CFG or if  $K$ is complete and completely idempotent, then $h(L) \in \CF(\Sigma,K)$.
\end{lm}

\begin{proof} Let $\M =(Q,\Gamma,q_0,\gamma_0,F,T)$ be a PDA over
  $\Delta$ with $L(\M) =L$. Moreover, by Theorem \ref{th:PDA=CF}, if
  $L = L(\G)$ for some unambiguous CFG $\G$, then we can assume that $\M$ is unambiguous.
 Let $\overline{\delta} \in \Delta$ be an arbitrary, but fixed element. 

\sloppy  The following construction employs the same technique as in
\cite[Lemma 5.7]{drovog12} of coding the preimage of $h$ into the state
set; thereby non-injectivity of $h$ is handled appropriately. We construct the PDA over $\Sigma$ with weights $\M' =
(Q',\Gamma,q_0',\gamma_0,F',T',\wt)$ where $Q' = \{q_0'\} \cup  Q \times
(\Delta \cup \{\varepsilon\})$ for some $q_0' \not\in Q$, $F' = F\times \{\overline{\delta}\}$, and $T'$ and $\wt$ are defined as follows.

\begin{itemize}
\item For every $x \in \Delta \cup \{\varepsilon\}$, the rule $\tau =
  (q_0',\varepsilon,\gamma_0,(q_0,x),\gamma_0)$ is in $T'$ and
  $\wt(\tau) =1$.

\item Let $\tau = (q,x,\gamma,p,\pi) \in T$ and $x' \in \Delta\cup
\{\varepsilon\}$. 
\begin{itemize}
\item If $x \in \Delta$ and $h(x) = a.y$, then $\tau' =
  ((q,x),y,\gamma,(p,x'),\pi) \in T'$  and $\wt(\tau') = a$.

\item If $x = \varepsilon$, then $\tau' =
  ((q,\varepsilon),\varepsilon,\gamma,(p,x'),\pi) \in T'$ and
  $\wt(\tau') = 1$.
\end{itemize}
\end{itemize}

Let $w \in \Sigma^*$. First, let $v \in \Delta^*$ with $h(v) =z.w$ for some $z \in K$. We write $v = \delta_1 \ldots \delta_n \in \Delta^*$ with  $n\ge 0$ and
$\delta_i \in \Delta$.  Let $h(\delta_i) = a_i.y_i$ for every
$1 \le i \le n$. Thus $h(v) = \Val(a_1,\ldots,a_n). y_1\ldots y_n$ and $w = y_1 \ldots y_n$.

 Let $\theta = \tau_1 \ldots \tau_m$ be a $q_0$-computation in $\Theta_{\M}(v)$; note that $m \ge \max\{n,1\}$ because at least $\gamma_0$ has to be popped. Let $x_i$ be the second component of $\tau_i$, so, $x_i \in \Delta \cup \{\varepsilon\}$, and $v = x_1 \ldots x_m$. 

Then we construct the  $q_0'$-computation $\theta' = \tau_0'\tau_1'\ldots \tau_m'$ in $\Theta_{\M'}(y_1\ldots y_n)$  as follows:
\begin{itemize}
\item $\tau_0' = (q_0',\varepsilon,\gamma_0,(q_0,x_1),\gamma_0)$.

\item If $1 \le i \le m$ and  $\tau_i = (q,x_i,\gamma,p,\pi)$, then $\tau_i' =
  ((q,x_i),y',\gamma,(p,x_{i+1}),\pi)$ where $y' = y$ if $x_i \in \Delta$ and
  $h(x_i) = a.y$, and $y' = \varepsilon$ if $x_i = \varepsilon$, and $x_{m+1}=\overline{\delta}$.

 \end{itemize}

Note that if $x_i \in \Delta$ and $h(x_i) = a.y$, then $\wt(\tau_i') =
a$, and if $x_i = \varepsilon$, then $\wt(\tau_i') = 1$ for each $1
\le i \le m$, by definition of $\wt$. Consequently
\[
(h(v),w) = \Val(a_1,\ldots,a_n) =
\Val(\wt(\tau_0'),\wt(\tau_1'),\ldots,\wt(\tau_m')) = \wt(\theta').
\]
In particular, $\wt(\theta_1') = (h(v),w) = \wt(\theta_2')$ for every $\theta_1,\theta_2 \in \Theta_{\M}(v)$.

Conversely, for every  $q_0'$-computation $\theta' =
\tau_0'\tau_1'\ldots \tau_m'$ in $\Theta_{\M'}(w)$ by definition of $T'$ there are a uniquely determined $v \in \Delta^*$ and a uniquely determined  $q_0$-computation $\theta = \tau_1 \ldots \tau_m$ in $\Theta_{\M}(v)$ such that $\theta'$ is the computation constructed above. It follows that $\M'$ is a WPDA.
 
So, for every $w \in \Sigma^*$ we obtain 
\begin{align*}
  (h(L(\M)),w) \; & = \; \left( \sum\nolimits_{v \in L(\M)} h(v), \, w\right) \\[3mm]
& = \;  \sum\nolimits_{\substack{v \in L(\M):\\(h(v),w) \not=0}} (h(v), w)\\[3mm]
& =^*   \sum\nolimits_{\substack{v \in L(\M), \theta \in \Theta_{\M}(v):\\(h(v),w) \not=0}} \wt(\theta')\\[2mm]
& = \;  \sum\nolimits_{\theta' \in \Theta_{\M'}(w)} \wt(\theta')\\[2mm]
& = \; (\beh \M'\beh, w)
\end{align*}
where the $*$-marked equality holds because (1) $K$ is complete and
completely idempotent or (2) $\M$ is unambiguous. 
Thus $\beh \M' \beh = h(L(\M))$ and the result follows from Theorem \ref{th:PDA=CF}.
\end{proof}

The following simple observation can be easily proved by considering \cite{chosch63} and  using \cite{barpersha61}, respectively. 

\newpage

\begin{lm}\label{ob:Dyck}
\begin{enumerate}
\item There is an unambiguous CFG $\G$ such that $L(\G) = \D_Y$.
\item Let $\G$ be an unambiguous CFG over $\Sigma$ and $R
  \subseteq \Sigma^*$ is a recognizable language. Then there is an
  unambiguous CFG $\G'$ with $L(\G') = L(\G) \cap R$. 
\end{enumerate}
\end{lm}

\begin{proof} 1. We consider the CFG $\G_Y' =
(N,Y \cup \overline{Y}, Z,P)$ with $N = \{Z,A\}$ and the rules $Z
\rightarrow AZ$, $Z \rightarrow
\varepsilon$, and  $A\rightarrow y Z \overline{y}$. Clearly, $L(\G_Y')
= \D_Y$, and $\G_Y'$ is unambiguous (cf. \cite[Prop.1, p.145]{chosch63}). 

2. Using \cite[Lm.4.1]{barpersha61} we can construct the CFG $\HH_1$
such that $L(\HH_1) = L(\G) \setminus \{\varepsilon\}$ and $\HH_1$ does
not contain a production of the form $A \rightarrow \varepsilon$. By
inspection of the construction we obtain that also $\HH_1$ is unambiguous. 

Now let ${\cal A}$ be a deterministic finite-state string automaton
such that $L({\cal A}) = R$. Then we can apply the construction of
\cite[Thm.8.1]{barpersha61} to $\HH_1$ and ${\cal A}$ and obtain an unambiguous CFG $\HH_2$ with $L(\HH_2) = L(\HH_1) \cap R$. Let $Z$ be the initial nonterminal of $\HH_2$. Finally, from $\HH_2$ we can construct the unambiguous CFG $\G'$ as follows:
if $\varepsilon \in L(\G) \cap R$, then we add the new  initial
nonterminal $Z'$ and the  productions $Z' \rightarrow Z$ and  $Z'
\rightarrow \varepsilon$ to $\HH_2$, otherwise let $\G' = \HH_2$. It
is clear that $\G'$ is unambiguous and $L(\G') = L(\G) \cap R$.
\end{proof}


As consequence of  Lemmas  \ref{th:mezwri} and \ref{lm:closure} and of the result of  Chomsky-Sch\"utzenberger for context-free languages  we can now derive the result of  Chomsky-Sch\"utzenberger for quantitative context-free languages.

\

{\em Proof of Theorem \ref{th:char}:} $1 \Leftrightarrow 3$: Immediate by Lemmas \ref{th:mezwri} and \ref{lm:closure}. 

$2 \Rightarrow 3$: by Observation \ref{ob:Dyck}.

$3 \Rightarrow 2$: By the classical result of Chomsky-Sch\"utzenberger (cf. e.g. \cite[Thm.G1]{koz97})
there are an alphabet $Y$, a recognizable  language $R$ over $Y \cup
\overline{Y}$, and an alphabetic morphism $g: Y \cup \overline{Y}
\rightarrow \Delta \cup \{\varepsilon\}$  such that $L(\G) = g(D_Y
\cap R)$. By analysis of the construction, we have that the set $g^{-1}(v) \cap D_Y \cap R$ is in a one-to-one correspondence with $D_\G(v)$, for every $v \in L(\G)$. Since $\G$ is unambiguous, we have that $|g^{-1}(v) \cap D_Y \cap R| = 1$.

Next we prove that $(h\circ g(v') \mid v' \in \D_Y \cap R)$ is locally
finite in case $K$ is not complete. Since $h(L(\G))$ is defined, the set $I_w = \{v \in L(\G)
\mid (h(v),w) \not= 0\}$ is finite for every $w \in \Sigma^*$. Since
$g^{-1}(v)\cap D_Y \cap R$ is a singleton for every $v \in L(\G)$, the set $\{v' \in \D_Y \cap R \mid (h(g(v')),w)  \not= 0\}$
is finite for every $w \in \Sigma^*$. Hence $(h\circ g(v') \mid v' \in
\D_Y \cap R)$ is locally finite.

Thus $h \circ g: Y \cup \overline{Y} \rightarrow  K\langle\Sigma \cup \{\varepsilon\}\rangle$ is an alphabetic morphism,  $(h \circ g)(D_Y \cap R)$ is well defined, and $s = (h \circ g)(D_Y \cap R)$.

$2 \Rightarrow 4$: Let $\widetilde{Y} = Y \cup \overline{Y}$.  Recall that $h(v) \in K\langle \Sigma^*\rangle$ is a monome for every $v \in \widetilde{Y}^*$. We define $h': \widetilde{Y}^* \rightarrow \Sigma^*$ by letting $h'(v) = w$ if $h(v) = a.w$. Clearly, $h'$ is a morphism. 

Choose a deterministic finite automaton ${\cal A}' = (Q,q_0,F,T)$ over $\widetilde{Y}$ recognizing $R$. We define a deterministic WFA ${\cal A} = (Q,q_0,F,T,\wt)$ over $\widetilde{Y}$ by putting $\wt(t) = a$ if $t = (q,z,p)$ and $h(z) = a.x$. Let $r = \beh {\cal A} \beh \in K\fl \widetilde{Y}^* \fr$. Note that 
$(r,v) = (h(v), w)$  if $v \in R$ and $h'(v) =w$, and 
$(r,v) = 0$ otherwise.

By assumption $s = h(D_Y \cap R) = \sum_{v \in D_Y \cap R} h(v)$. Hence, for $w \in \Sigma^*$:
\[ (s,w) 
= \sum_{v \in D_Y \cap R} (h(v),w)
=  \sum_{v \in D_Y, h'(v)=w} (r,v)
=  (h'(r \cap D_Y),w)
\]
where the sums exist because they have only finitely many nonzero entries if $K$ is not complete.
Thus $s = h'(r \cap D_Y)$.  

$4 \Rightarrow 3$: We put $\widetilde{Y} = Y \cup \overline{Y}$. Also,
let $\widetilde{Y}_0 = \widetilde{Y}  \cup \{\gamma_0\}$ with an element
$\gamma_0 \not\in \widetilde{Y}$. By assumption, there is a deterministic
WFA ${\cal A} = (Q,q_0,F,T,\wt)$ over $\widetilde{Y}$ with $\beh {\cal
  A} \beh = r$. We let ${\cal A}' = (Q,q_0,F,T)$, a deterministic
finite automaton over $\widetilde{Y}$. 


Next, we wish to define a PDA
${\cal M}$ over $\widetilde{Y}$ recognizing $L({\cal A}') \cap
\D_Y$. Let ${\cal M} = (Q,\widetilde{Y}_0,q_0,\gamma_0,F,T')$  such that 
\begin{itemize}
\item $(q,\varepsilon,\gamma_0,q,\varepsilon) \in T'$ for each $q \in F$, and
\item for every $x \in \widetilde{Y}$ and $\gamma \in \widetilde{Y}_0$, $(q,x,\gamma,p,\pi) \in T'$ iff $(q,x,p)
  \in T$ and 
\[
\pi = 
\left\{
\begin{array}{ll}
x\gamma & \text{ if } x \in Y\\
\varepsilon & \text{ if } \gamma \in Y \text{ and } x = \overline{\gamma}.
\end{array}
\right.
\]
\end{itemize}
Since ${\cal A}'$ is deterministic, ${\cal M}$ is an unambiguous
PDA. Moreover, we have $L({\cal M}) = L({\cal A}') \cap \D_Y$. 

Next, we extend ${\cal M}$ to a PDA ${\cal M}_T =
(Q,\widetilde{Y}_0,q_0,\gamma_0,F,\overline{T})$ over $T$ by letting 
\[
\begin{array}{rcl}
(q,t,\gamma,p,\pi) \in \overline{T} & \text{iff} & (q,x,\gamma,p,\pi)
\in T' \text{ and }\\
&& \text{either } t = (q,x,p) \in T \text{ or } t = x = \varepsilon. 
\end{array}
\]
Clearly, ${\cal M}_T$ is unambiguous. Moreover, since ${\cal A}'$ is
deterministic, for each $v \in L({\cal M}) \subseteq L({\cal A}')$
there is a unique successful path $p_v \in T^*$ on $v$ in ${\cal
  A}'$. Then $p_v \in L({\cal M}_T)$. Conversely, each $v' \in L({\cal M}_T)$ arises as $v' = p_v$
for a uniquely determined word $v \in L({\cal M})$ in this way.

We let $\lab: T^* \rightarrow \widetilde{Y}^*$ be the alphabetic
morphism mapping each transition to its label, i.e., $\lab(q,x,p) =
x$. Finally we define an alphabetic morphism $h_K: T \rightarrow
K\langle \Sigma \cup \{\varepsilon\}\rangle$ by letting $h_K(t) = \wt(t).h(\lab(t))$.

We claim that $h_K(L({\cal M}_T)) = h(r \cap \D_Y)$. Let $w \in
\Sigma^*$. Note that if $v \in L({\cal M})$ and $v' = p_v$ as above,
then $\lab(v') = v$ and $h_K(v') = \wt(v').h(v)$. Since $v' = p_v$ is
the unique successful path in ${\cal A}$ on $v$, we obtain $\wt(v') =
(\beh {\cal A}\beh,v)$. Moreover, $(h_K(v'),w) \not= 0$ implies $w =
h(v)$. Also, $(\beh {\cal A}\beh,v) = 0$ if $v \not\in L({\cal
  A}')$. Hence we obtain
\[
\begin{array}{l}
(h_K(L({\cal M}_T)),w) 
= \sum_{v' \in L({\cal M}_T)} (h_K(v'),w)
= \sum_{\substack{v \in L({\cal M})\\h(v)=w}} \wt(v')\\
= \sum_{\substack{v \in L({\cal A}')\cap\D_Y\\h(v)=w}} (\beh {\cal   A}\beh,v) 
= \sum_{\substack{v \in \D_Y\\h(v)=w}} (\beh {\cal   A}\beh,v) 
= \sum_{\substack{v \in \widetilde{Y}^*\\h(v)=w}} (r \cap D_Y,v) \\[5mm]
= (h(r \cap \D_Y), w).
\end{array}
\]

\

$3 \Rightarrow 5$: trivial.

$5 \Rightarrow 1$: by Lemma \ref{lm:closure}.
\hfill $\Box$


\section{Context-Free Step Functions}
\label{sec:step-functions}

An important result in the theory of rational power series
states sufficient conditions when for a recognizable series,
the language of all words assuming a given value is recognizable.
Of particular interest are semirings in which both operations
are restricted to be locally finite \cite{berreu88},
cf. also \cite{drostuvog10}. Here, as a supplement of the previous results,
we investigate this question for quantitative context-free languages over unital valuation
monoids $K$. This can also be seen as another way (in comparison
to Theorem \ref{th:char}) of connecting series with languages and weights.

If $L\subseteq\Sigma^*$ and $a\in K$, we let $a\cdot\one_L\in K\fl\Sigma^*\fr$ be the series satisfying $(a\cdot\one_L,w)=a$ if $w\in L$, and $(a\cdot\one_L,w)=0$ otherwise. 
Let $s\in K\fl\Sigma^*\fr$. We say that $s$ is a \emph{context-free step function} if
$s = \sum_{i=1}^n a_i \cdot \one_{L_i}$ for some $n \in \nat$,
$a_1,\ldots,a_n \in K$, and context-free languages $L_1,\ldots,L_n
\subseteq \Sigma^*$. The languages $L_i$ are called {\em step
  languages}. Moreover, a context-free step function is a {\em
  recognizable step function} if each $L_i$ is recognizable.

As is well known, a recognizable step function over any semiring is a
recognizable series \cite{eil74}. This even holds for strong bimonoids with the same proof \cite{drostuvog10}, and could be extended to unital valuation monoids. In contrast, this implication fails for context-free languages and quantitative context-free languages.

\begin{lm}\label{ex:inh-amb}  Let $L$ be an inherently ambiguous context-free
  language. Then $\one_L~\not\in~\CF(\Sigma,\mathbb{N})$. 
\end{lm}
\begin{proof} Assume that $\one_L \in \CF(\Sigma,\mathbb{N})$. Then let $\G$ be a WCFG such that   $\one_L = \beh \G \beh$. Then, for every word $w \in L$, we have   that $|D_\G(w)| =1$, because otherwise the weights of different
  derivations of $w$ would sum up to a value greater than 1. But then
  the CFG which is underlying $\G$ is an unambiguous CFG for $L$,
  which contradicts our assumption on $L$. 
\end{proof}    

However under suitable restrictions we obtain the  following positive result.

\begin{lm}\label{lm:1=cfstep} Let $L$ be a context-free language over
  $\Sigma$  and $a \in K$. If $L$   can be generated by some
  unambiguous CFG or if $K$ is complete and completely idempotent, then $a \cdot \one_L \in \CF(\Sigma,K)$.
\end{lm}

\begin{proof} We choose a new symbol $\# \not\in \Sigma$ and define
  the context-free language $L' = \{\#w \mid w \in L\}$ over $\Sigma'
  = \Sigma \cup \{\#\}$. Moreover, we define the alphabetic morphism
  $h: \Sigma' \rightarrow K\langle \Sigma \cup \{\varepsilon\}\rangle$
  by $h(\#) = a.\varepsilon$ and $h(\sigma) = 1.\sigma$ for every
  $\sigma \in \Sigma$. 

Clearly, if $L$ can be generated by some unambiguous CFG, then this
also holds for $L'$. Moreover, $(h(v) \mid v\in L')$ is locally finite
in case $K$ is not complete.  It follows that $a\cdot \one_L = h(L')$.  By Lemma
\ref{lm:closure} we obtain that $h(L') \in \CF(\Sigma,K)$.
\end{proof}

We call a context-free step function a {\em context-free step function
  with unambiguous step languages} if each of its step languages can be
generated by an unambiguous CFG. Then we obtain the following result by Lemmas \ref{lm:1=cfstep} and  \ref{lm:WPDA-sum} and Theorem~\ref{th:PDA=CF}.

\begin{cor}\label{cor:CF-QCF} Let $s \in K\fl \Sigma^*\fr$ be a context-free step
  function. Let $s$ have unambiguous step languages or let $K$ be
  complete and completely idempotent. Then $s \in \CF(\Sigma,K)$.
\end{cor}

A context-free step function $s$ is called \emph{strong} if it can be expressed as $\sum_{i=1}^n a_i \cdot \one_{L_i}$
 where the family $(L_i  \mid 1 \le i \le n)$ forms a partition of $\Sigma^*$, i.e.,  $L_i \cap L_j = \emptyset$ for every $i \not= j$, and $\Sigma^* = \bigcup_{i=1}^n L_i$. Clearly,  $s$
is a strong context-free step function iff  $\im(s)$ is finite and $s^{-1}(a)$ is a context-free language for every $a \in \im(s)$, where $\im(s) = \{(s,w)\mid w \in \Sigma^*\}$ is the \emph{image} of $s$.

Due to the closure properties of the class of recognizable languages,
we can transform every recognizable step function over $\Sigma^*$ into an equivalent
one for which the collection of its step languages partitions
$\Sigma^*$. This is different for context-free step functions.

\begin{lm}\rm There are quantitative context-free languages which are context-free step
functions  but  not strong context-free step functions.
\end{lm}
\begin{proof} For this
consider, e.g., the context-free step function $s = 1\cdot \one_{L_1}
+ 2 \cdot \one_{L_2}$ over the semiring of natural numbers with $L_1 =
\{a^nb^nc^k \mid n,k \in \nat\}$ and $L_2 = \{a^kb^nc^n \mid n,k \in
\nat\}$. (Recall from Example \ref{ex:val} that we can view any semiring as a particular unital valuation monoid.). Since $L_1$ and
$L_2$ can be generated by unambiguous CFG, we obtain from Corollary
\ref{cor:CF-QCF} that $s$ is a quantitative context-free languages. But $3 \in \im(s)$ and  $s^{-1}(3) = \{a^nb^nc^n \mid n \in \nat\}$ which is not context-free. 
\end{proof}

In the rest of this section we wish to derive a converse of Corollary \ref{cor:CF-QCF}. However, even if the unital valuation monoid $K$ is finite, complete, and completely idempotent, then the converse in general does not hold, as shown by the following example.

\begin{ex} Let  $\Sigma = \{\sigma\}$, and choose a non recursively enumerable set $L \subseteq \Sigma^*$ with $\sigma \not\in L$. We define $K$ with $|K| =3$, say, $K = \{0,1,a\}$, as additively idempotent unital valuation monoid such that 
$\Val(a^n) = 1$ if $\sigma^n \in L$, and $\Val(a^n) = a$ if $\sigma^n \not\in L$, for each $n \ge 2$. There are no other restrictions on $+$ or $\Val$.

Let ${\cal A} = (Q,q_0,F,T,\wt)$ over $\Sigma$ with $Q = F = \{q_0\}$, $T = \{(q_0,\sigma,q_0)\}$, and $\wt((q_0,\sigma,q_0)) =a$. Clearly ${\cal A}$ is a deterministic WFA and 
$(\beh {\cal A} \beh,w) = 
\Val(a^{|w|}) = 
1$  if $w \in L$  or $w = \varepsilon$, and $(\beh {\cal A} \beh,w) =  a$ otherwise. 
So $\beh {\cal A}\beh$ takes on two values but $\beh {\cal A}\beh^{-1}(1)= L \cup \{\varepsilon\}$ which is not recursively enumerable.
\end{ex}

The example shows that we have to impose additional assumptions on the valuation function $\Val$. We call a unital valuation monoid $K$ {\em sequential} if 
\[\Val(a_1,\ldots,a_{n+1}) = \Val(\Val(a_1,\ldots,a_n),a_{n+1})\]
for every $n \ge 1$ and  $a_1,\ldots,a_{n+1} \in K$. 
Given a sequential unital valuation monoid $K$, we can define a multiplication $\cdot$ on $K$ by letting $a \cdot b = \Val(a,b)$ for $a,b \in K$. Clearly, then $(K,+,\cdot,0,1)$ is a unital monoid-magma (compare Example \ref{ex:val}).
Hence sequential unital valuation monoids are precisely the valuation monoids arising from unital monoid-magmas.

We say that $K$ is \emph{locally finite}, if whenever $F\subseteq K$ is a finite subset, then the set $\Val(F^*)$ comprising the valuations of all finite sequences of elements of $F$ is finite.

\begin{ex} \sloppy We let  $\mathbb{R}^t$ (truncated reals) denote the set of all
real numbers of a given bounded precision. We consider the unital valuation monoid $K=(\mathbb{R}^t \cup \{-\infty,\infty\}, \sup,\Val, -\infty,\infty)$ where  $\Val$ arises, as described above,  from the operation $\cdot$ defined as binary average followed by truncation to a number in  $\mathbb{R}^t$. Then $K$ is infinite, completely idempotent, and locally finite. Also, $\cdot$ is not associative, since if $\mathbb{R}^t = \mathbb{Z}$ we have $(1 \cdot 5) \cdot 9 = 6$  but $1 \cdot (5 \cdot 9) = 4$.
This describes a situation where we may have successive average computations only of sufficient bounded precision. 

For  examples of locally finite strong bimonoids we refer the reader to \cite{drovog12}.
\end{ex}


Now we show:

\begin{theo}\label{lm:WPDA-cfstep} Let $K$ be an additively
  idempotent, locally finite, sequential unital valuation monoid. Assume
  that $K$ is completely idempotent in case that $K$ is complete. Then
  each series $s \in \CF(\Sigma,K)$ is a context-free step function.
\end{theo}
\begin{proof} Due to Theorem \ref{th:PDA=CF}, there is a WPDA $\M = (Q,\Gamma,q_0,\gamma_0,F,T,\wt)$ over $\Sigma$ and $K$ such that $\beh \M \beh = s$. Clearly, $\wt(T)$ is a finite subset of $K$. Let $Y=\Val(\wt(T)^*)$. Then $Y$ is  finite. Also note that $\wt(\theta) \in Y$ for every computation $\theta$ of $\M$.

Let $w \in \Sigma^*$. Then
\[
(\beh \M\beh,w) = \sum_{\theta \in \Theta(w)} \wt(\theta) = \sum_{a \in Y}  \sum_{\substack{\theta \in \Theta(w)\\\wt(\theta) = a}} \wt(\theta) =^{*} \sum_{a \in Y} a \cdot \one_{L_a}(w)
\]
where $L_a = \{v \in \Sigma^* \mid \text{there is a $\theta \in
  \Theta(v)$ with $\wt(\theta)=a$}\}$. The equation marked by $*$
holds because $K$ is  additively idempotent or  completely idempotent in case $K$ is complete. It remains to prove that $L_a$ is context-free for every $a \in Y$.

We construct the PDA $\M_a = (Q',\Gamma,q_0',\gamma_0,F',T')$ where $Q' = Q \times Y$, $q_0' = (q_0,1)$, and $F' = F \times \{a\}$. Moreover, we let 
$T'$ contain $((q,y),x,\gamma,(p,y'),\pi)$ iff $\tau = (q,x,\gamma,p,\pi) \in T$ and $y'=\Val(y,\wt(\tau))$.
 Then $q_0'$-computations of $\M_a$ correspond to $q_0$-computations $\theta$ of $\M$ with $\wt(\theta)=a$, and conversely. It follows that $L(\M_a)= L_a$, hence $L_a$ is a context-free language.
\end{proof}

Note that  in the previous theorem, in general $s$ is not a strong context-free step function, because $L_a$ and $L_{a'}$ need not be disjoint.

A lattice is called {\em complete} if any subset has a supremum and
infimum. Clearly, every complete lattice is a completely
idempotent and  locally finite sequential unital valuation monoid. So, as a consequence of
Corollary \ref{cor:CF-QCF} and Theorem \ref{lm:WPDA-cfstep}, we obtain immediately the following.

\begin{cor}\label{cor:bounded-lattices} Let $K$ be a complete lattice and $s \in  K\fl\Sigma^*\fr$. Then $s \in \CF(\Sigma,K)$ if and only if $s$ is a context-free step function.
\end{cor}

\section{Conclusion and Open Problems}

We could show that a fundamental result of formal language theory, the Chomsky-Sch\"utzenberger theorem, holds not only in semiring-weighted settings, but even for much more general computation models, the unital valuation monoids, which include, e.g., computations of averages of real numbers. 
We can represent the quantitative languages by weighted context-free grammars and by weighted pushdown automata; both formalisms were shown to be expressively equivalent. 
Finally, we considered context-free step functions.

In \cite{drostuvog10} it was proved that every recognizable series over any additively locally finite and multiplicatively locally finite strong bimonoid is a recognizable step function. In the light of this result, we wonder whether it is possible to extend our Theorem \ref{lm:WPDA-cfstep} to additively locally finite strong bimonoids (while keeping the restriction on the multiplication operation).

Recently, a (unweighted) Chomsky-Sch\"utzenberger Theorem has been proved in which the  involved morphism is non-erasing \cite{okh12}. Can this be generalized to our weighted setting?


\end{document}